\documentclass[twoside]{article}%
\usepackage{amssymb}
\usepackage{amsfonts}
\usepackage{amsmath}
\usepackage{graphicx}%
\providecommand{\U}[1]{\protect\rule{.1in}{.1in}}
\newtheorem{theorem}{Theorem}

\newtheorem{corollary}[theorem]{Corollary}

\newtheorem{lemma}[theorem]{Lemma}

\newtheorem{remark}[theorem]{Remark}

\newenvironment{proof}[1][Proof]{\noindent\textbf{#1.} }{\ \rule{0.5em}{0.5em}}
\ifx\pdfoutput\relax\let\pdfoutput=\undefined\fi
\newcount\msipdfoutput
\ifx\pdfoutput\undefined\else
\ifcase\pdfoutput\else
\ifx\paperwidth\undefined\else
\ifdim\paperheight=0pt\relax\else\pdfpageheight\paperheight\fi
\ifdim\paperwidth=0pt\relax\else\pdfpagewidth\paperwidth\fi
\fi\fi\fi
\begin{document}

\title{\textbf{Analytical} \textbf{Blowup Solutions to the }$2$\textbf{-dimensional
Isothermal Euler-Poisson Equations of Gaseous Stars II}}
\author{Y\textsc{uen} M\textsc{anwai\thanks{E-mail address: nevetsyuen@hotmail.com }}\\\textit{Department of Applied Mathematics, }\\\textit{The Hong Kong Polytechnic University,}\\\textit{Hung Hom, Kowloon, Hong Kong}}
\date{Revised 29-July-2009}
\maketitle

\begin{abstract}
This article is the continued version of the analytical blowup solutions for
2-dimensional Euler-Poisson equations \cite{Y1}. With extension of the blowup
solutions with radial symmetry for the isothermal Euler-Poisson equations in
$R^{2}$, other special blowup solutions in $R^{2}$ with non-radial symmetry
are constructed by the separation method. We notice that the results are the
first evolutionary solutions with non-radial symmetry for the system.


\end{abstract}

\section{Introduction}

The evolution of a self-gravitating fluid (gaseous stars) can be formulated by
the isentropic Euler-Poisson equations of the following form:
\begin{equation}
\left\{
\begin{array}
[c]{rl}%
{\normalsize \rho}_{t}{\normalsize +\nabla\cdot(\rho\vec{u})} &
{\normalsize =}{\normalsize 0,}\\
{\normalsize (\rho\vec{u})}_{t}{\normalsize +\nabla\cdot(\rho\vec{u}%
\otimes\vec{u})+\nabla P} & {\normalsize =}{\normalsize -\rho\nabla\Phi,}\\
{\normalsize \Delta\Phi(t,\vec{x})} & {\normalsize =\alpha(N)}%
{\normalsize \rho,}%
\end{array}
\right.  \label{Euler-Poisson}%
\end{equation}
where $\alpha(N)$ is a constant related to the unit ball in $R^{N}$:
$\alpha(1)=2$; $\alpha(2)=2\pi$ and For $N\geq3,$%
\begin{equation}
\alpha(N)=N(N-2)V(N)=N(N-2)\frac{\pi^{N/2}}{\Gamma(N/2+1)},
\end{equation}
where $V(N)$ is the volume of the unit ball in $R^{N}$ and $\Gamma$ is the
Gamma function. And as usual, $\rho=\rho(t,\vec{x})$ and $\vec{u}=\vec
{u}(t,\vec{x})=(u_{1},u_{2},....,u_{N})\in\mathbf{R}^{N}$ are the density and
the velocity respectively.

In the above system, the self-gravitational potential field $\Phi=\Phi
(t,\vec{x})$\ is determined by the density $\rho$ through the Poisson equation.

The equation (\ref{Euler-Poisson})$_{3}$ is the Poisson equation through which
the gravitational potential is determined by the density distribution of the
density itself. Thus, we call the system (\ref{Euler-Poisson}) the
Euler-Poisson equations. The equations can be viewed as a prefect gas model.
The function $P=P(\rho)$\ is the pressure. The $\gamma$-law can be applied on
the pressure $P(\rho)$, i.e.%
\begin{equation}
{\normalsize P}\left(  \rho\right)  {\normalsize =K\rho}^{\gamma}%
:=\frac{{\normalsize \rho}^{\gamma}}{\gamma}, \label{gamma}%
\end{equation}
which is a commonly the hypothesis. The constant $\gamma=c_{P}/c_{v}\geq1$,
where $c_{P}$, $c_{v}$\ are the specific heats per unit mass under constant
pressure and constant volume respectively, is the ratio of the specific heats,
that is, the adiabatic exponent in (\ref{gamma}). In particular, the fluid is
called isothermal if $\gamma=1$. It can be used for constructing models with
non-degenerate isothermal cores, which have a role in connection with the
so-called Schonberg-Chandrasekhar limit \cite{KW}.

The system can be rewritten as%
\begin{equation}
\left\{
\begin{array}
[c]{rl}%
\rho_{t}+\nabla\cdot\vec{u}\rho+\nabla\rho\cdot\vec{u} & {\normalsize =}%
{\normalsize 0,}\\
\rho\left(  \frac{\partial u_{i}}{\partial t}+\sum_{k=1}^{N}u_{k}%
\frac{\partial u_{i}}{\partial x_{k}}\right)  +\frac{\partial}{\partial x_{i}%
}P(\rho) & {\normalsize =}{\normalsize -\rho}\frac{\partial}{\partial x_{i}%
}{\normalsize \Phi(\rho),}\text{ for }i=1,2,...N,\\
{\normalsize \Delta\Phi(t,x)} & {\normalsize =\alpha(N)}{\normalsize \rho.}%
\end{array}
\right.  \label{eqeq1}%
\end{equation}
For $N=3$, (\ref{eqeq1}) is a classical (non-relativistic) description of a
galaxy, in astrophysics. See \cite{BT}, \cite{C} and \cite{KW} for a detail
about the system.

For the local existence results about the system were shown in \cite{M},
\cite{B} and \cite{G}.\ In particular, the radially symmetric solutions can be
expressed by%
\begin{equation}
\rho(t,\vec{x})=\rho(t,r)\text{ and }\vec{u}(t,\vec{x})=\frac{\vec{x}}%
{r}V(t,r):=\frac{\vec{x}}{r}V,
\end{equation}
where the radial diameter $r:=\left(  \sum_{i=1}^{N}x_{i}^{2}\right)  ^{1/2}$.
Historically in astrophysics, Goldreich and Weber constructed the analytical
blowup (collapsing) solutions of the $3$-dimensional Euler-Poisson equations
for $\gamma=4/3$ for the non-rotating gas spheres \cite{GW}. After that,
Makino \cite{M1} obtained the rigorously mathematical proof of the existence
of such kind of blowup solutions. Besides, Deng, Xiang and Yang extended the
above blowup solutions in $R^{N}$ ($N\geq3$) \cite{DXY}. Then, Yuen obtained
the blowup solutions in $R^{2}$ with $\gamma=1$ by a new transformation
\cite{Y1}. The family of the analytical solutions are rewritten as

For $N\geq3$ and $\gamma=(2N-2)/N$, in \cite{DXY}
\begin{equation}
\left\{
\begin{array}
[c]{c}%
\rho(t,r)=\left\{
\begin{array}
[c]{c}%
\dfrac{1}{a^{N}(t)}y(\frac{r}{a(t)})^{N/(N-2)},\text{ for }r<a(t)Z_{\mu};\\
0,\text{ for }a(t)Z_{\mu}\leq r.
\end{array}
\right.  \text{, }V{\normalsize (t,r)=}\dfrac{\dot{a}(t)}{a(t)}%
{\normalsize r,}\\
\ddot{a}(t){\normalsize =}\dfrac{-\lambda}{a^{N-1}(t)},\text{ }%
{\normalsize a(0)=a}_{1}\neq0{\normalsize ,}\text{ }\dot{a}(0){\normalsize =a}%
_{2},\\
\ddot{y}(z){\normalsize +}\dfrac{N-1}{z}\dot{y}(z){\normalsize +}\dfrac
{\alpha(N)}{(2N-2)K}{\normalsize y(z)}^{N/(N-2)}{\normalsize =\mu,}\text{
}y(0)=\alpha>0,\text{ }\dot{y}(0)=0,
\end{array}
\right.  \label{solution2}%
\end{equation}
where $\mu=[N(N-2)\lambda]/(2N-2)K$ and the finite $Z_{\mu}$ is the first zero
of $y(z)$;

For $N=2$ and $\gamma=1$, in \cite{Y1}%
\begin{equation}
\left\{
\begin{array}
[c]{c}%
\rho(t,r)=\dfrac{1}{a^{2}(t)}e^{y\left(  r/a(t)\right)  }\text{,
}V{\normalsize (t,r)=}\dfrac{\dot{a}(t)}{a(t)}{\normalsize r;}\\
\ddot{a}(t){\normalsize =}\dfrac{-\lambda}{a(t)},\text{ }{\normalsize a(0)=a}%
_{1}>0{\normalsize ,}\text{ }\dot{a}(0){\normalsize =a}_{2};\\
\ddot{y}(z){\normalsize +}\dfrac{1}{z}\dot{y}(z){\normalsize +\dfrac
{\alpha(2)}{K}e}^{y(z)}{\normalsize =\mu,}\text{ }y(0)=\alpha,\text{ }\dot
{y}(0)=0,
\end{array}
\right.  \label{solution 3}%
\end{equation}
where $K>0$, $\mu=2\lambda/K$ with a sufficiently small $\lambda$ and $\alpha$
are constants.

And for other special blowup solutions, the readers may see the details in
\cite{Y}.

Very recently, Yuen extended the above solutions to the pressureless
Navier-Stokes-Poisson equations with density-dependent viscosity in
\cite{Y2}\textit{:}%
\begin{equation}
\left\{
\begin{array}
[c]{rl}%
\rho_{t}+V\rho_{r}+\rho V_{r}+{\normalsize \dfrac{N-1}{r}\rho V} &
{\normalsize =0,}\\
\rho\left(  V_{t}+VV_{r}\right)  +\dfrac{\alpha(N)\rho}{r^{N-1}}%
{\displaystyle\int_{0}^{r}}
\rho(t,s)s^{N-1}ds & {\normalsize =}[\kappa\rho^{\theta}]_{r}\left(
\frac{N-1}{r}V+V_{r}\right)  +(\kappa\rho^{\theta})(V_{rr}+\dfrac{N-1}{r}%
V_{r}+\dfrac{N-1}{r^{2}}V).
\end{array}
\right.
\end{equation}
However, the known solutions are all in radial symmetry. In this paper, we are
able to obtain the similar results to the non-radially symmetric cases for the
$2$-dimensional Euler-Poisson equations (\ref{eqeq1}) in the following theorem:

\begin{theorem}
\label{thm2 copy(2)}For the $2$-dimensional isothermal Euler-Poisson equations
(\ref{eqeq1}), there exists a family of solutions,%
\begin{equation}
\left\{
\begin{array}
[c]{c}%
\rho(t,x,y)=\frac{1}{a(t)^{2}}e^{-\frac{\Phi\left(  \frac{Ax+By}{a(t)}\right)
}{K}+C}\text{, }{\normalsize \vec{u}(t,x,y)=\dfrac{\dot{a}(t)}{a(t)}%
(}x,y){\normalsize ,}\\
a(t)=a_{1}+a_{2}t,\\
\ddot{\Phi}(s)-\epsilon^{\ast}e^{-\frac{\Phi(s)}{K}}=0,\text{ }\Phi
(0)=\alpha,\text{ }\dot{\Phi}(0)=\beta,
\end{array}
\right.  \label{ss2}%
\end{equation}
where $A,$ $B,$ ${\normalsize a}_{1}\neq0,$ ${\normalsize a}_{2\text{ }},$
$\frac{2\pi e^{C}}{A^{2}+B^{2}}=\epsilon^{\ast}>0$ with that $A$ and $B$ are
not both $0$, $\alpha$ and $\beta$ are constants.\newline In particular,
$a_{1}>0$ and $a_{2}<0$, the solutions (\ref{ss2}) blow up in the finite time
$T=-a_{2}/a_{1}$.
\end{theorem}

\begin{remark}
The solutions (\ref{ss2}) is a line source or sink in terms of cylindrical
coordinates, when we take $z$ direction to lie alone the characteristic line
of the source or sink. For the physical significance of such kind of
solutions, the interested readers may refer.P.409-410 of \cite{JK} for details.
\end{remark}

\section{Separable Blowup Solutions}

Before presenting the proof of Theorem \ref{thm2 copy(2)}, we prepare some
lemmas first.

\begin{lemma}
\label{lem:generalsolutionformasseq copy(1)}For the continuity equation
(\ref{eqeq1})$_{1}$ in $R^{2}$, there exist solutions,%
\begin{equation}
\rho(t,x,y)=\frac{f\left(  \dfrac{Ax+By}{a(t)}\right)  }{a^{2}(t)},\text{
}{\normalsize \vec{u}(t,x,y)=\frac{\overset{\cdot}{a}(t)}{a(t)}(}x,y),
\end{equation}
where the scalar function $f(s)\geq0\in C^{1}$ and $a(t)\neq0\in C^{1}.$
\end{lemma}

\begin{proof}
We plug the solutions (\ref{ss2}) into the continuity equation (\ref{eqeq1}%
)$_{1}$,
\begin{align}
&  \rho_{t}+\nabla\cdot\vec{u}\rho+\nabla\rho\cdot\vec{u}\\
&  =\frac{\partial}{\partial t}\left[  \frac{f\left(  \dfrac{Ax+By}%
{a(t)}\right)  }{a(t)^{2}}\right]  +\nabla\cdot\frac{\dot{a}(t)}%
{a(t)}(x,y)\frac{f\left(  \dfrac{Ax+By}{a(t)}\right)  }{a(t)^{2}}+\nabla
\frac{f\left(  \dfrac{Ax+By}{a(t)}\right)  }{a(t)^{2}}\cdot\frac{\dot{a}%
(t)}{a(t)}(x,y)\\
&  =\frac{-2\dot{a}(t)}{a(t)^{3}}f\left(  \frac{Ax+By}{a(t)}\right)  +\frac
{1}{a(t)^{2}}\frac{\partial}{\partial t}f\left(  \frac{Ax+By}{a(t)}\right)
+\frac{\dot{a}(t)}{a(t)}\left(  \frac{\partial}{\partial x}x+\frac{\partial
}{\partial y}y\right)  \frac{f\left(  \dfrac{Ax+By}{a(t)}\right)  }{a(t)^{2}%
}\\
&  +\frac{\dot{a}(t)}{a(t)}\left[  \frac{\partial}{\partial x}\frac{f\left(
\dfrac{Ax+By}{a(t)}\right)  }{a(t)^{2}}\cdot x+\frac{\partial}{\partial
y}\frac{f\left(  \dfrac{Ax+By}{a(t)}\right)  }{a(t)^{2}}\cdot y\right]  \\
&  =\frac{-2\dot{a}(t)}{a(t)^{3}}f\left(  \frac{Ax+By}{a(t)}\right)  -\frac
{1}{a(t)^{2}}\dot{f}\left(  \frac{Ax+By}{a(t)}\right)  \frac{(Ax+By)\dot
{a}(t)}{a(t)^{2}}+2\frac{\dot{a}(t)}{a(t)}\frac{f\left(  \dfrac{Ax+By}%
{a(t)}\right)  }{a(t)^{2}}\\
&  +\frac{\dot{a}(t)}{a(t)}\left[  \frac{\dot{f}\left(  \dfrac{Ax+By}%
{a(t)}\right)  }{a(t)^{2}}\frac{Ax}{a(t)}+\frac{\dot{f}\left(  \dfrac
{Ax+By}{a(t)}\right)  }{a(t)^{2}}\frac{By}{a(t)}\right]  \\
&  =0.
\end{align}
The proof is completed.
\end{proof}

On the other hand, we may use the fixed point theorem to show the local
existence of the following ordinary differential equation:

\begin{lemma}
\label{lemma2}There exists a sufficiently small $x_{0}>0$, such that the
equation%
\begin{equation}
\left\{
\begin{array}
[c]{c}%
\ddot{\Phi}(s)-\epsilon^{\ast}e^{-\tfrac{\Phi(s)}{K}}=0,\\
\Phi(0)=\alpha\text{, }\dot{\Phi}(0)=\beta,
\end{array}
\right.  \label{SecondorderElliptic}%
\end{equation}
where $\epsilon^{\ast}>0$, $\alpha$, and $\beta$ are constants, has a solution
$\Phi=\Phi(s)\in C^{2}[0,s_{0}]$.
\end{lemma}

\begin{proof}
We integrate the equation (\ref{SecondorderElliptic}) once:
\begin{equation}
\overset{\cdot}{\Phi}(s)=-\int_{0}^{s}K\epsilon^{\ast}e^{-\frac{\Phi(\eta)}%
{K}}d\eta+\beta.
\end{equation}
And set%
\begin{equation}
{\normalsize f(s,\Phi(s))=-}\int_{0}^{s}{\normalsize K\epsilon}^{\ast
}e^{-\frac{\Phi(\eta)}{K}}{\normalsize d\eta+\beta.}%
\end{equation}
then for any $s_{0}>0$, we get $f\in C^{1}[0,$ $s_{0}]$. and for any
$\Phi_{1,}$ $\Phi_{2}\in C^{2}[0,$ $s_{0}]$, we have,%
\begin{equation}
\left\vert f(s,\Phi_{1}(s))-f(s,\Phi_{2}(s))\right\vert =K\epsilon^{\ast
}\left\vert \int_{0}^{s}(e^{-\frac{\Phi_{1}(\eta)}{K}}-e^{-\frac{\Phi_{2}%
(\eta)}{K}})d\eta\right\vert .
\end{equation}
As $e^{y}$ is a $C^{1}$ function of $y$, we can show that the function $e^{y}%
$, is Lipschitz-continuous. And we get,%
\begin{equation}
\left\vert f(s,\Phi_{1}(s))-f(s,\Phi_{2}(s))\right\vert =O(1)\int_{0}%
^{s}\left\vert \left(  \Phi_{1}(\eta\right)  -\Phi_{2}(\eta)\right\vert
d\eta\leq O(1)s_{0}\underset{0\leq s\leq s_{0}}{\sup}\left\vert \Phi
_{1}(s)-\Phi_{2}(s)\right\vert .
\end{equation}
We let%
\begin{equation}
{\normalsize T\Phi(s)=\alpha+}\int_{0}^{s}{\normalsize f(\eta,\Phi(\eta
))d\eta.}%
\end{equation}
We have $T\Phi\in C[0,$ $s_{0}]$\ and%
\begin{equation}
\left\vert T\Phi_{1}(s)-T\Phi_{2}(s)\right\vert =\left\vert \int_{0}^{s}%
f(\eta,\Phi_{1}(\eta))d\eta-\int_{0}^{s}f(\eta,\Phi_{2}(\eta))d\eta\right\vert
\leq O(1)s_{0}\underset{0\leq s\leq s_{0}}{\sup}\left\vert \Phi(s)_{1}%
-\Phi(s)_{2}\right\vert .
\end{equation}
By choosing $s_{0}>0$ to be a sufficiently small number, such that
$O(1)s_{0}<1$, this shows that the mapping $T:C[0,$ $s_{0}]\rightarrow C[0,$
$s_{0}]$, is a contraction with the sup-norm. By the fixed point theorem,
there exists a unique $\Phi(s)\in C[0,$ $s_{0}],$\ such that $T\Phi
(s)=\Phi(s)$. The proof is completed.
\end{proof}

And we need another lemma to show the global existence of the ordinary
differential equation (\ref{SecondorderElliptic}).

\begin{lemma}
\label{lemma3}The equation,%
\begin{equation}
\left\{
\begin{array}
[c]{c}%
\ddot{\Phi}(s)-\epsilon^{\ast}e^{-\tfrac{\Phi(s)}{K}}=0,\\
\Phi(0)=\alpha\text{, }\dot{\Phi}(0)=\beta,
\end{array}
\right.  \label{Elliptic1}%
\end{equation}
where $\epsilon^{\ast}>0$, $\alpha$ and $\beta$ are constants, has a solution
in $(-\infty,$ $+\infty)$ and $\underset{s\rightarrow\pm\infty}{\lim}%
\Phi(s)=+\infty$.
\end{lemma}

\begin{proof}
We prove the case $\dot{\Phi}(0)=0$ first. By integrating (\ref{Elliptic1}),
we have,%
\begin{equation}
\overset{\cdot}{\Phi}(s)=\epsilon^{\ast}\int_{0}^{s}e^{-\tfrac{\Phi(\eta)}{K}%
}d\eta\geq0.\label{lemma3eq1}%
\end{equation}
Thus, for $0<s<s_{0}$, $\Phi(s)$ has a uniform lower bound
\begin{equation}
\Phi(s)\geq\Phi(0)=\alpha.
\end{equation}
As we obtained he local existence in Lemma \ref{lemma2}, there are two
possibilities:\newline(1)$\Phi(s)$ only exists in some finite interval $[0,$
$s_{0}]$: (1a)$\underset{s\rightarrow s_{0-}}{\lim}\Phi(s)=+\infty$;
(1b)$\Phi(s)$ has an uniformly upper bound, i.e. $\Phi(s)\leq\alpha_{0}$ for
some constant $\alpha_{0}.$\newline(2)$\Phi(s)$ exists in $[0,$ $+\infty)$:
(2a)$\underset{s\rightarrow+\infty}{\lim}\Phi(s)=+\infty$; (2b)$\Phi(s)$ has
an uniformly lower bound, i.e. $\Phi(s)\leq\alpha$ for some constant
$\alpha_{0}$.\newline We claim that possibility (1) does not exist. We need to
reject (1b) first: If the statement (1b) is true, (\ref{lemma3eq1}) becomes%
\begin{equation}
\epsilon^{\ast}se^{-\tfrac{\alpha_{0}}{K}}\geq\epsilon^{\ast}\int_{0}%
^{s}e^{-\tfrac{\Phi(\eta)}{K}}d\eta=\overset{\cdot}{\Phi}(s).\label{possible1}%
\end{equation}
Thus, $\overset{\cdot}{\Phi}(s)$ is bounded in $[0,s_{0}]$. Therefore, we can
use the fixed point theorem again to obtain a large domain of existence, such
that $[0,s_{0}+\delta]$ for some positive number $\delta$. There is a
contradiction. Therefore, (1b) is rejected.\newline Next, we do not accept
(1a) because of the following reason: It is impossible that $\underset
{s\rightarrow s_{0-}}{\lim}\Phi(s)=+\infty$, as from (\ref{possible1}),
$\overset{\cdot}{\Phi}(x)$ has a upper bound in $[0,$ $s_{0}]$:%
\begin{equation}
\epsilon^{\ast}se^{-\tfrac{\alpha_{0}}{K}}\geq\overset{\cdot}{\Phi
}(s).\label{lemma3eq2}%
\end{equation}
Thus, (\ref{lemma3eq2}) becomes,
\begin{align}
\Phi(s_{0}) &  =\Phi(0)+\epsilon^{\ast}\int_{0}^{s_{0}}\overset{\cdot}{\Phi
}(\eta)d\eta\\
&  \leq\alpha+\epsilon^{\ast}\int_{0}^{s_{0}}se^{-\tfrac{\alpha_{0}}{K}}%
d\eta\\
&  =\alpha+\epsilon^{\ast}s_{0}^{2}e^{-\tfrac{\alpha_{0}}{K}}.
\end{align}
Since $\Phi(s)$ is bounded upper in $[0,$ $s_{0}]$, it contracts the statement
(1a), such that $\underset{s\rightarrow s_{0-}}{\lim}\Phi(s)=+\infty$. So, we
can exclude the possibility (1).\newline We claim that the possibility (2b)
does not exist. It is because
\begin{equation}
\overset{\cdot}{\Phi}(s)=\epsilon^{\ast}\int_{0}^{s}e^{-\tfrac{\Phi(\eta)}{K}%
}d\eta\geq\epsilon^{\ast}\int_{0}^{s}e^{-\tfrac{\alpha_{0}}{K}}d\eta
=\epsilon^{\ast}e^{-\tfrac{\alpha_{0}}{K}}s.
\end{equation}
Then, we have,%
\begin{equation}
\Phi(s)\geq\alpha+\epsilon^{\ast}e^{-\tfrac{\alpha_{0}}{K}}s.\label{lemma3eq3}%
\end{equation}
By letting $s\rightarrow+\infty$, (\ref{lemma3eq3}) turns out to be,
\begin{equation}
\Phi(s)=+\infty.
\end{equation}
Since a contradiction is established, we exclude the possibility (2b). Thus,
the equation (\ref{Elliptic1}) exists in $[0,$ $+\infty)$ and $\underset
{s\rightarrow+\infty}{\lim}\Phi(s)=+\infty$. Due to the solution is symmetric
about $s=0$, we have $\underset{s\rightarrow-\infty}{\lim}\Phi(s)=+\infty$.

For the case $\dot{\Phi}(0)=\beta$, these exist constants $\tilde{s}$ and
$\alpha_{\tilde{s}}$, such that $\Phi(0)=\beta$ in the ordinary differential
equation:
\begin{equation}
\left\{
\begin{array}
[c]{c}%
\ddot{\Phi}(s)-\epsilon^{\ast}e^{-\tfrac{\Phi(s)}{K}}=0,\\
\Phi(\tilde{s})=\alpha_{\tilde{s}}\text{, }\dot{\Phi}(\tilde{s})=0.
\end{array}
\right.  \label{fgf}%
\end{equation}
The above fact is due to the transformation $z=s-\tilde{s}$  to have%
\begin{equation}
\left\{
\begin{array}
[c]{c}%
\ddot{\Phi}(z)-\epsilon^{\ast}e^{-\tfrac{\Phi(z)}{K}}=0,\\
\Phi(0)=\alpha_{\tilde{s}}\text{, }\dot{\Phi}(0)=0.
\end{array}
\right.
\end{equation}
Then we get $\underset{z\rightarrow\pm\infty}{\lim}\Phi(z)=\underset
{s\rightarrow\pm\infty}{\lim}\Phi(s)=+\infty$ with the previous result and the
continuity of the solutions (\ref{fgf}) to show the lemma is true. This
completes the proof.
\end{proof}

On the other hand, the following lemma handles the Poisson equation
(\ref{eqeq1})$_{3}$ for our solutions (\ref{ss2}):

\begin{lemma}
\label{lemma2 copy(1)}The solutions,
\begin{equation}
\rho=\frac{{\normalsize 1}}{a(t)^{2}}e^{-\frac{\Phi\left(  \frac{Ax+By}%
{a(t)}\right)  }{K}+C}, \label{aass1}%
\end{equation}
with the second-order ordinary differential equation:%
\begin{equation}
\ddot{\Phi}(s)-\epsilon^{\ast}e^{-\tfrac{\Phi(s)}{K}}=0,\text{ }\Phi
(0)=\alpha\text{, }\dot{\Phi}(0)=\beta,
\end{equation}
where $s:=(Ax+By)/a(t)$ and $C$, $\frac{2\pi e^{C}}{A^{2}+B^{2}}%
=\epsilon^{\ast}>0$, $\alpha$ and $\beta$ are constants,\newline fit into the
Poisson equation (\ref{eqeq1})$_{3}$ in $R^{2}$.
\end{lemma}

\begin{proof}
We check that our potential function $\Phi(t,x,y)$ satisfies the Poisson
equation (\ref{eqeq1})$_{3}$:%
\begin{align}
&  {\normalsize \Delta\Phi(t,x,y)-2\pi}{\normalsize \rho}\\
&  =\nabla\cdot\nabla\Phi\left(  \frac{Ax+By}{a(t)}\right)  -\frac
{{\normalsize 2\pi}}{a(t)^{2}}e^{-\frac{^{\Phi\left(  \frac{Ax+By}%
{a(t)}\right)  }}{K}+C}\\
&  =\nabla\cdot\left[  \dot{\Phi}\left(  \frac{Ax+By}{a(t)}\right)  \frac
{A}{a(t)},\dot{\Phi}\left(  \frac{Ax+By}{a(t)}\right)  \frac{B}{a(t)}\right]
-\frac{{\normalsize 2\pi}}{a(t)^{2}}e^{-\frac{\Phi\left(  \frac{Ax+By}%
{a(t)}\right)  }{K}+C}\\
&  =\frac{\partial}{\partial x}\left[  \dot{\Phi}(\frac{Ax+By}{a(t)})\frac
{A}{a(t)}\right]  +\frac{\partial}{\partial y}\left[  \dot{\Phi}(\frac
{Ax+By}{a(t)})\frac{B}{a(t)}\right]  -\frac{{\normalsize 2\pi}}{a(t)^{2}%
}e^{-\frac{\Phi\left(  \frac{Ax+By}{a(t)}\right)  }{K}+C}\\
&  =\frac{A^{2}+B^{2}}{a(t)^{2}}\left(  \ddot{\Phi}(s)-\frac{2\pi e^{C}}%
{A^{2}+B^{2}}e^{-\tfrac{\Phi(s)}{K}}\right)  ,
\end{align}
where $A$ and $B$ are not both $0$.\newline Then, we choose $s:=(Ax+By)/a(t)$
and the ordinary differential equation:%
\begin{equation}
\ddot{\Phi}(s)-\epsilon^{\ast}e^{-\tfrac{\Phi(s)}{K}}=0,\text{ }\Phi
(0)=\alpha\text{, }\dot{\Phi}(0)=\beta,
\end{equation}
with $\frac{2\pi e^{C}}{A^{2}+B^{2}}=\epsilon^{\ast}$, $\alpha$ and $\beta$
are constants in Lemmas \ref{lemma2} and \ref{lemma3}. Therefore, our
solutions (\ref{aass1}) satisfy the Poisson equation (\ref{eqeq1})$_{3}$.

The proof is completed.
\end{proof}

Now, we are ready to check that the solutions fit into the Euler-Poisson
equations (\ref{eqeq1}).

\begin{proof}
[Proof of Theorem \ref{thm2 copy(2)}]By Lemma
\ref{lem:generalsolutionformasseq copy(1)} and Lemma \ref{lemma2 copy(1)}, the
solutions (\ref{ss2}) satisfy (\ref{eqeq1})$_{1}$ and (\ref{eqeq1})$_{3}$. For
the $x$-component of the isothermal momentum equations (\ref{eqeq1})$_{2}$ in
$R^{2}$, we have%
\begin{align}
&  \rho\left(  \frac{\partial u_{1}}{\partial t}+u_{1}\frac{\partial u_{1}%
}{\partial x}+u_{2}\frac{\partial u_{1}}{\partial y}\right)  +\frac{\partial
}{\partial x}K\rho+\rho\frac{\partial\Phi}{\partial x}\\
&  =\rho\left[  \frac{\partial}{\partial t}\left(  \frac{\dot{a}(t)}%
{a(t)}x\right)  +\frac{\dot{a}(t)}{a(t)}x\frac{\partial}{\partial x}\left(
\frac{\dot{a}(t)}{a(t)}x\right)  +\frac{\dot{a}(t)}{a(t)}y\frac{\partial
}{\partial y}\left(  \frac{\dot{a}(t)}{a(t)}x\right)  \right] \\
&  +K\frac{\partial}{\partial x}\frac{e^{-\frac{\Phi\left(  \frac{Ax+By}%
{a(t)}\right)  }{K}+C}}{a(t)^{2}}+\rho\dot{\Phi}\left(  \frac{Ax+By}%
{a(t)}\right)  \frac{A}{a(t)}\\
&  =\rho\left[  \frac{\ddot{a}(t)}{a(t)}x\right]  -\rho\dot{\Phi}\left(
\frac{Ax+By}{a(t)}\right)  \frac{A}{a(t)}+\rho\dot{\Phi}\left(  \frac
{Ax+By}{a(t)}\right)  \frac{A}{a(t)}\\
&  =0,
\end{align}
by taking $\ddot{a}(t)=0$ that is%
\begin{equation}
a(t)=a_{1}+a_{2}t.
\end{equation}
For the $y$-component of the isothermal momentum equations (\ref{eqeq1})$_{2}%
$, the proof is similar. As the readers may check it, we omit the detail
here.\newline We have shown that the solutions (\ref{ss2}) satisfy the
Euler-Poisson equations. In particular, $a_{1}>0$ and $a_{2}<0$, the solutions
(\ref{ss2}) blow up in the finite time $T=-a_{2}/a_{1}$.

The proof is completed.
\end{proof}

\begin{remark}
For the case of $A=0$ and $B=0$, the corresponding solutions (\ref{ss2}) may
be deduced to the special solutions in \cite{Y}.
\end{remark}

\begin{remark}
Our solutions (\ref{ss2}) also work for the isothermal Navier-Stokes-Poisson
equations in $R^{2}$:
\begin{equation}
\left\{
\begin{array}
[c]{rl}%
{\normalsize \rho}_{t}{\normalsize +\nabla\cdot(\rho\vec{u})} &
{\normalsize =}{\normalsize 0,}\\
{\normalsize (\rho\vec{u})}_{t}{\normalsize +\nabla\cdot(\rho\vec{u}%
\otimes\vec{u})+\nabla K\rho} & {\normalsize =}{\normalsize -\rho\nabla
\Phi+\mu\Delta\vec{u},}\\
{\normalsize \Delta\Phi(t,x)} & {\normalsize =2\pi}{\normalsize \rho,}%
\end{array}
\right.
\end{equation}
where $\mu>0$ is a positive constant.
\end{remark}

Additionally, the blowup rate about the solutions is immediately followed:

\begin{corollary}
\label{thm:2 copy(2)}The blowup rate of the solutions (\ref{ss2}) is,%
\begin{equation}
\underset{t\rightarrow T}{\lim}\rho(t,0,0)\left(  T-t\right)  ^{2}\geq O(1).
\end{equation}

\end{corollary}

In conclusion, due to the novel solutions obtained by the separation method,
the author conjectures there exists other analytical solution in non-radial
symmetry. Further works will be continued for seeking more particular
solutions to understand the nature of the Euler-Poisson equations (\ref{eqeq1}).

\end{document}